\documentclass[11pt]{article}

\usepackage[T1]{fontenc}
\usepackage{lmodern}
\usepackage[margin=1in]{geometry}
\usepackage{amsmath,amssymb,amsthm,mathtools}
\usepackage{booktabs,longtable,array,tabularx}
\usepackage{enumitem}
\usepackage{xcolor}
\usepackage{microtype}
\usepackage{xurl}
\usepackage[colorlinks=true,linkcolor=blue!55!black,citecolor=blue!55!black,urlcolor=blue!55!black]{hyperref}

\newtheorem{theorem}{Theorem}[section]
\newtheorem{proposition}[theorem]{Proposition}

\theoremstyle{definition}

\theoremstyle{remark}
\newtheorem{remark}[theorem]{Remark}

\newcommand{\R}{\mathbb R}
\newcommand{\conv}{\operatorname{conv}}

\newcommand{\eps}{\varepsilon}
\newcommand{\NS}{\operatorname{NS}}

\setlist{itemsep=2pt,topsep=5pt}
\title{Counterexamples to GNS Conjecture 1.8}
\author{Xiuqing Duan\\
  {\small School of Physical and Mathematical Sciences, Nanyang Technological University, Singapore}\\
  {\small\texttt{xiuqing.duan@ntu.edu.sg}}}
\date{}

\begin{document}
\maketitle

\begin{abstract}
We disprove both parts of Conjecture 1.8 of Gabrielov, Novikov, and Shapiro (GNS). 
For the Coulomb potential, a generic rational configuration of 24 unit charges in $\mathbb{R}^3$ 
has at least 18 nondegenerate critical points of Morse index one but exactly 14 effective Voronoi one-cells. 
A separate proper-line construction has at least two nondegenerate minima but only one 
relatively effective zero-dimensional cell intersection. For each positive subspace dimension,
product suspension gives a proper-subspace counterexample.
\end{abstract}

\tableofcontents

\section{Statement, conventions, and conclusions}

Let $A_1,\dots,A_\ell\in\R^n$ be distinct sites.  In the notation of
Gabrielov--Novikov--Shapiro (GNS)~\cite{GNS}, the potential is
\begin{equation}\label{eq:potential}
 V_\alpha(x)=\sum_{i=1}^{\ell}\zeta_i\rho_i(x)^{-\alpha},
\qquad\text{where}\qquad \rho_i(x)=|x-A_i|^2,
 \qquad \alpha>0.
\end{equation}
The charge sites are excluded from the domain.  Conjecture 1.8 in~\cite{GNS} concerns
positive unit charges, so $\zeta_i=1$.  If $a_\alpha^j$ denotes the number of
critical points of the potential with $j$ negative Hessian directions, and $\sharp^j$ the number of effective
Voronoi cells of dimension $j$ in the Voronoi diagram of the considered
configuration, part (a) of Conjecture 1.8 asserts that, for every generic configuration of unit point charges,
\begin{equation}\label{eq:ambient-conj}
 a_\alpha^j\leq\sharp^j\qquad(\alpha\geq1/2).
\end{equation}
Part (b) of Conjecture 1.8 asserts, for an affine subspace $L$ generically intersecting the
diagram,
\begin{equation}\label{eq:relative-conj}
 a_{\alpha,L}^j\leq\sharp_L^j,
\end{equation}
where the Hessian and index are restricted to the tangent bundle $T L$, and $\sharp_L^j$ counts
cells $S$ for which $\dim(S\cap L)=j$ and which are effective relative to
$L$.  The index ranges are $0\leq j\leq n$ and $0\leq j\leq\dim L$,
respectively. The printed part (b) does not restate the range of $\alpha$.
Our counterexamples use $\alpha=1/2$, so they refute both the reading
$\alpha>0$ and the inherited range $\alpha\geq1/2$.

Edelsbrunner--Fillmore--Oliveira (EFO)~\cite{EFO} write
$V_p=\sum_i|x-A_i|^{-p}$.  For unit charges,
\begin{equation}\label{eq:conversion}
 p=2\alpha,
\end{equation}
so EFO's Coulomb case $p=1$ is exactly GNS $\alpha=1/2$.

Our main result is as follows:
\begin{theorem}[]\label{thm:main}
Conjecture 1.8(a), as printed, is false.  Conjecture 1.8(b) is also false,
independently of whether its ambiguous containment notation admits
$L=\R^n$, because a proper-line counterexample is given below.  For each
integer $d\geq1$, there exists a configuration of unit charges in $\R^{d+2}$
and a proper affine $d$-plane $L_d$ meeting its Voronoi diagram generically
for which the relative inequality fails.  Under the standard zero-manifold
Morse convention, the relative inequality is true for $\dim L=0$.
\end{theorem}

The proof is given in Sections~\ref{sec:ambient}--\ref{sec:relative}. The
certifying scripts and deterministic machine-readable outputs accompany this
paper in \path{research/certification/}.

\subsection{Exact Voronoi terminology}

A Voronoi cell $S$ is the relatively open set of points having
exactly the same nearest-site set $\NS(S)$. Its Delaunay cell is $\conv\NS(S)$.  GNS
call $S$ effective when
\begin{equation}\label{eq:effective}
 S\cap\conv\NS(S)\ne\varnothing.
\end{equation}
GNS call a configuration generic when every codimension-$k$ cell has
exactly $k+1$ nearest sites and avoids the boundary of $\conv\NS(S)$~\cite[p.~446]{GNS}. Here boundary is taken relative to the
affine span of that convex hull. This interpretation is necessary for
consistency: a full-dimensional Voronoi cell contains its singleton dual, whose relative boundary is empty.

For an affine subspace $L$, let $\pi_L$ be orthogonal projection.  A cell is
effective relative to $L$ when
\begin{equation}\label{eq:relative-effective}
 S\cap\conv\bigl(\pi_L\NS(S)\bigr)\ne\varnothing.
\end{equation}
The pair of the sites and $L$ is generic when $L$ meets every encountered
cell transversely, an encountered codimension-$k$ cell has $k+1$ nearest
sites, and the intersection avoids the relative boundary of the projected
hull.

GNS separately introduce the condition that all critical points be
nondegenerate.  It is not part of their definition of a generic site
configuration.  All critical points counted below are nondegenerate, and
Section~\ref{sec:global-morse} also treats the stronger global-Morse reading.

\subsection{What GNS Theorem 1.7 does prove}

After normalizing the common sign to positive, GNS Theorem 1.7 gives, for
each fixed generic configuration, a configuration-dependent $\alpha_0$ such
that the effective-cell
correspondence and index equality hold for $\alpha\geq\alpha_0$.  The proof
in their Sections 2.4--2.5 is a large-exponent localization argument.  It
does not supply a configuration-independent threshold or monotonicity down
to $1/2$.  The printed theorem says ``same sign,'' but for all-negative
charges multiplication of the potential by $-1$ preserves critical
locations while complementing ordinary Morse indices; its index statement
therefore requires normalization of the common sign to positive \cite[Section~2.4]{GNS}.

The printed statement of Theorem~1.7(b) lists only positive-codimension
ambient cells. In Section~2.5, however, GNS also include a restricted
local maximum associated with each effective full-dimensional cell
whose site lies off $L$; if the site lies on $L$, its projection is
excluded as a charge singularity. This discrepancy does not affect
either counterexample.

\section{The symmetric EFO configuration}\label{sec:efo}

Up to similarity, the vertices of the regular truncated octahedron used by
EFO are
\begin{equation}\label{eq:trunc-sites}
 \mathcal T=\{\text{signed coordinate permutations of }(0,1,2)\}.
\end{equation}
The body-centered-cubic Voronoi normalization is $\mathcal T/4$.  Scaling
preserves critical counts and Hessian indices.

Every $a\in\mathcal T$ has $|a|^2=5$, and
\begin{equation}
 |x-a|^2=|x|^2+5-2x\cdot a.
\end{equation}
Thus the nearest sites at $x$ form the exposed face of the truncated
octahedron maximizing $x\cdot a$.  Its face lattice consists of one
three-cell, 14 facets, 36 edges, and 24 vertices.  Exact normal-cone
intersections show that all are effective, giving
\begin{equation}\label{eq:symm-counts}
 (\sharp^0,\sharp^1,\sharp^2,\sharp^3)=(1,14,36,24).
\end{equation}

The four observed noncentral critical orbits have representatives and
indices shown in Table~\ref{tab:efo-orbits}.  These exploratory values
independently reproduce EFO's 18 index-one and 36 index-two observations;
they are not used as certificates for the new theorem.

\begin{table}[ht]
\centering
\caption{Reconstructed noncentral EFO orbits in normalization
\eqref{eq:trunc-sites}.}\label{tab:efo-orbits}
\begin{tabular}{@{}lrr@{}}
\toprule
Representative & Orbit size & Index\\
\midrule
$(1.494626559931,0,0)$ & 6 & 1\\
$(0.626750041258,0.626750041258,0)$ & 12 & 1\\
$(1.198510470821,1.198510470821,0)$ & 12 & 2\\
$(1.656799953819,0.338792038314,0.338792038314)$ & 24 & 2\\
\bottomrule
\end{tabular}
\end{table}

The origin is critical by central symmetry.  Since
\begin{equation}
 \sum_{a\in\mathcal T}aa^T=40I,
\end{equation}
its Hessian is exactly
\begin{equation}
 D^2V(0)=\frac{3}{5^{5/2}}40I-\frac{24}{5^{3/2}}I=0.
\end{equation}
The center is degenerate.  
EFO present this configuration as a counterexample \cite[Result~3 and Section~4.2]{EFO}. They report 18 index-one and 36 index-two equilibria
and indicate that equilibria on suitable symmetry lines can be established by the method illustrated for the cube in their Appendix~A.
The configuration is not GNS-generic, since the origin is a zero-dimensional Voronoi
cell with all 24 sites nearest, while GNS genericity in $\R^3$ requires four. EFO do not supply a perturbation
argument establishing the required genericity together with the effective-cell count. The next section gives an explicit rational
perturbation and verifies both the effective-cell count and the critical-point lower bound.

\section{An explicit generic counterexample to part (a)}\label{sec:ambient}

\subsection{Rational site coordinates}

We use EFO's truncated-octahedron configuration as a starting point.
Resolving its Voronoi degeneracies may increase the effective-cell
count, so the perturbation was selected by an exploratory search
for a small effective-cell count and then checked exactly.
The resulting configuration has 14 effective one-cells and
at least 18 certified index-one critical points.

Let $D=10^6$ and $A_i=P_i/D$, where the integer numerator vectors are listed
in Table~\ref{tab:sites}.  They are a small, nonsymmetric rational
perturbation of \eqref{eq:trunc-sites}.  Put a positive unit charge at every
site and take $\alpha=1/2$.

\begin{table}[ht]
\centering
\caption{Exact numerator vectors $P_i$ for the ambient counterexample. The left and right blocks list $P_1,\ldots,P_{12}$ and
$P_{13},\ldots,P_{24}$, respectively.}
\label{tab:sites}
\begin{tabular}{@{}l r r r @{\hspace{2.4em}} l r r r@{}}
\toprule
$P_i$ & $P_{i,1}$ & $P_{i,2}$ & $P_{i,3}$ & $P_i$ & $P_{i,1}$ & $P_{i,2}$ & $P_{i,3}$\\
\midrule
$P_{1}$&-2009085&-1005649&-9966 & $P_{13}$&-8309&1003557&-1991347\\
$P_{2}$&-2002386&4264&-1000270 & $P_{14}$&-5488&999768&1997623\\
$P_{3}$&-1991672&-513&991417 & $P_{15}$&-9821&2005855&-995940\\
$P_{4}$&-2007006&996017&-3734 & $P_{16}$&-1670&2006121&990484\\
$P_{5}$&-990323&-1993944&-4600 & $P_{17}$&1001417&-1991860&-4151\\
$P_{6}$&-993116&-2748&-2006549 & $P_{18}$&1008033&-2028&-2003354\\
$P_{7}$&-994179&6569&1991471 & $P_{19}$&1002729&6204&2003524\\
$P_{8}$&-991702&2006224&4074 & $P_{20}$&993695&2005154&-4101\\
$P_{9}$&-9686&-1999988&-996151 & $P_{21}$&1994157&-1001408&8489\\
$P_{10}$&257&-2003478&1005873 & $P_{22}$&1996376&2282&-1003636\\
$P_{11}$&9701&-998695&-2005029 & $P_{23}$&2003253&37&994282\\
$P_{12}$&-7545&-1002162&2000615 & $P_{24}$&2004841&991779&-3731\\
\bottomrule
\end{tabular}
\end{table}

The potential is
\begin{equation}\label{eq:coulomb}
 V(x)=\sum_{i=1}^{24}|x-A_i|^{-1}.
\end{equation}

\subsection{Exact Delaunay and effectiveness computation}

The verifier \path{exact_voronoi_certificate.py} uses only integers and
rational numbers.  It checks all $\binom{24}{4}=10626$ orientation
determinants and all $\binom{24}{5}=42504$ five-point in-sphere determinants.
None vanishes.  The minimum absolute raw values are
\begin{equation}\label{eq:predicate-margins}
 76710796042601,\qquad
 7646698260706286959940590,
\end{equation}
respectively.  After the $D^{-1}$ coordinate scaling these become
$7.6710796042601\cdot10^{-5}$ and
$7.646698260706286959940590\cdot10^{-6}$.

Every four-subset is tested against all remaining sites.  Exactly 67
tetrahedra have a strictly empty circumsphere.  Taking all faces gives the
complete Delaunay triangulation:
\begin{equation}\label{eq:delaunay-counts}
 67\ \text{tetrahedra},\quad156\ \text{triangles},\quad
 112\ \text{edges},\quad24\ \text{vertices}.
\end{equation}

For each simplex, the verifier solves the circumcenter and its barycentric
coordinates over $\mathbb Q$ and compares its squared distance with every
other site.  A Delaunay simplex and its dual Voronoi cell have complementary
orthogonal affine spans; their only possible meeting is the circumcenter.
Therefore exact barycentric and nearest-site signs decide effectiveness.
They give
\begin{equation}\label{eq:generic-counts}
 (\sharp^0,\sharp^1,\sharp^2,\sharp^3)=(1,14,36,24).
\end{equation}
The 14 effective one-cells are dual to the zero-based triangles
\begin{align*}
 &(0,1,2),(0,4,5),(2,4,11),(2,6,7),(3,12,14),(4,9,16),\\
 &(5,12,17),(6,11,18),(7,14,19),(10,16,21),(11,18,20),\\
 &(12,19,21),(18,19,23),(20,21,23).
\end{align*}
All decisions are strict.  Two global strictness margins, taken over all 156 Delaunay triangles, are
\begin{equation}
 \min|\text{triangle barycentric coordinate}|=
 \frac{63982323834490170008}{1358453353382550052924069}>0,
\end{equation}
and
\begin{equation}\label{eq:triangle-margin}
 \min|\text{triangle nearest-site margin}|=
 \frac{3441751615753642410984725576857291}
 {12203706291298804970211929}>0.
\end{equation}
Here the margin in \eqref{eq:triangle-margin} is measured in
numerator-coordinate squared-distance units.  For the sites $A_i=P_i/D$, the
corresponding margin is $D^{-2}$ times the displayed value, with
$D^2=10^{12}$.
Thus every codimension-$k$ cell has $k+1$ nearest sites and no cell meets the
relative boundary of its Delaunay simplex.  The configuration is GNS-generic.

\subsection{Krawczyk certificates for 18 critical points}

For \eqref{eq:coulomb},
\begin{equation}\label{eq:gradient}
 F(x)=\nabla V(x)=-\sum_i\frac{x-A_i}{|x-A_i|^3},
\end{equation}
and
\begin{equation}\label{eq:hessian}
 H(x)=D^2V(x)=\sum_i\left(
 \frac{3(x-A_i)(x-A_i)^T}{|x-A_i|^5}
 -\frac{I}{|x-A_i|^3}\right).
\end{equation}
The machine-readable certificate contains 18 rational 75-digit centers
$c_s$.  Let
\begin{equation}
 X_s=c_s+[-10^{-35},10^{-35}]^3,\qquad1\leq s\leq18.
\end{equation}
At 256-bit precision, outward-rounded Arb arithmetic~\cite{Johansson} encloses $F(c_s)$ and
$[H](X_s)$.  For a rational point matrix $Y_s$, the verifier checks its
determinant exactly nonzero and evaluates
\begin{equation}\label{eq:krawczyk}
 K_s=c_s-Y_sF(c_s)+(I-Y_s[H](X_s))(X_s-c_s).
\end{equation}
The primary 256-bit Krawczyk test compares its image with
enclosures of the exact rational endpoints of $X_s$, proving
\begin{equation}\label{eq:krawczyk-inclusion}
 K_s\subset\operatorname{int}X_s.
\end{equation}
The independent 320-bit interval-Newton and 512-bit Krawczyk
verifiers described in Section~\ref{sec:audit} provide separate
exact-endpoint checks.
The Krawczyk theorem~\cite[Theorem~13.3]{Rump} yields a unique zero of $F$ in each box.  Exact rational
comparisons prove the boxes pairwise disjoint.  The minimum certified squared
distance from any box to a site exceeds $1.17$, and the maximum Krawczyk
image-radius/box-radius ratio is below $2.0\cdot10^{-32}$.

Interval evaluation also proves $\det H(X_s)<0$ for every $s$; the upper
endpoint closest to zero remains below $-7.39\cdot10^{-4}$.  Since
\eqref{eq:coulomb} is harmonic away from its sites,
$\operatorname{tr}H=0$ exactly.  A nonsingular symmetric $3\times3$ matrix
with trace zero and negative determinant has exactly one negative
eigenvalue.  All 18 critical points have Morse index one.

\begin{proof}[Proof of the first assertion of Theorem~\ref{thm:main}]
The interval boxes prove $a^1_{1/2}\geq18$.  The exact Voronoi certificate
proves $\sharp^1=14$ for the same rational sites.  Hence
\begin{equation}
 a^1_{1/2}\geq18>14=\sharp^1,
\end{equation}
contradicting \eqref{eq:ambient-conj}.  Additional critical points cannot
weaken this lower-bound violation.
\end{proof}

\subsection{An open chamber around the certified configuration}

Every orientation, in-sphere, barycentric, nearest-site, Krawczyk,
Hessian-determinant, and site-separation assertion above is a strict
inequality.  There are finitely many, and all depend continuously on the
sites on the certified boxes.  Therefore one open neighborhood of the
displayed rational site vector simultaneously preserves the 18 critical
points and indices, the complete Delaunay triangulation, every effectiveness
decision, and the strict inequality $18>14$.  No Delaunay flip or newly
effective triangle can occur there without a certified predicate vanishing.

\subsection{The stronger global-Morse reading}\label{sec:global-morse}

Even an added demand that every critical point be nondegenerate does not
rescue the conjecture.  Let $\mathcal P$ be the manifold of ordered distinct
site configurations and define, for $x\ne A_i$,
\begin{equation}
 \Phi(A,x)=\nabla_x\sum_i|x-A_i|^{-1}.
\end{equation}
For one site, the $x$-Hessian has radial and tangential eigenvalues
$2r_i^{-3},-r_i^{-3},-r_i^{-3}$.  Since
$D_{A_i}\Phi$ is its negative, it is invertible.  Thus $\Phi$ is a
submersion.  Parametric transversality makes the parameters for which
$\Phi_A$ is transverse to zero dense.

This global statement is safe on the punctured noncompact domain.  Near a
site, its self-field of size $r^{-2}$ dominates the uniformly bounded other
fields.  At infinity,
\begin{equation}
 \Phi(A,x)=-24x/|x|^3+O(|x|^{-3})
\end{equation}
uniformly in a bounded parameter neighborhood.  All critical points lie in a
common compact set separated from the sites.  A transverse fiber has finitely
many roots, and the all-Morse property is open by compactness and the
implicit function theorem.  Intersecting this dense open set with the open
certificate chamber gives a nonempty open set of all-Morse counterexamples;
it contains rational site vectors.

\section{The relative inequality}\label{sec:relative}

\subsection{The conditional full-space specialization}

GNS Conjecture 1.8(b) says ``any affine subspace'' and does not say ``proper.''
However, the preceding definition writes $L\subset\R^n$, whereas GNS use
$L\subseteq\R^n$ for the affine span of the sites on journal p.~450, and
they declare no consistent convention distinguishing the two glyphs.
Accordingly, the printed formulation does not explicitly clarify whether
the ambient space itself is admitted as $L$.

Under the inclusive convention, take $L=\R^n$.  Transversality is automatic,
projection is the identity, effectiveness is ordinary effectiveness,
$\dim(S\cap L)=\dim S$, and $V|_L=V$ with the same Hessian.  Therefore
\begin{equation}
 a_{\alpha,L}^j=a_\alpha^j,\qquad \sharp_L^j=\sharp^j.
\end{equation}
Under that convention, part (a) is the full-space instance of part (b).
Under a proper-subspace convention it is not.  In either interpretation,
part (b) is false by the independent proper-line counterexample below.

\subsection{A proper rational line}

The construction separates nearest-site geometry from the full
potential. We keep two opposite sites fixed and move the other
22 slightly away from the chosen line, without changing their
projections. The fixed pair then determines the nearest-site
diagram along the whole line, while all 24 charges contribute
to a restriction with two strict minima.

We now give an independent proper-subspace counterexample.  Let $\mathcal T$
be \eqref{eq:trunc-sites}, and set
\begin{equation}
 L=\{(0,t,t):t\in\R\},\qquad p=(0,1,2),\qquad\eps=\frac1{1000}.
\end{equation}
For $v=(x,y,z)$, let
\begin{equation}
 \pi_Lv=\left(0,\frac{y+z}{2},\frac{y+z}{2}\right),
 \qquad v^\perp=v-\pi_Lv,
\end{equation}
and define rational sites
\begin{equation}\label{eq:line-sites}
 B_v=\begin{cases}
 v,&v=\pm p,\\[2mm]
 \pi_Lv+\dfrac{1001}{1000}v^\perp,&v\ne\pm p.
 \end{cases}
\end{equation}
The 24 sites are distinct, centrally symmetric, and off $L$.

Put $s=y+z$.  Along $L$,
\begin{equation}
 q_v(t)=|(0,t,t)-B_v|^2=2t^2-2st+c_v.
\end{equation}
With
\begin{equation}
 \kappa=(1001/1000)^2-1=2001/10^6,
\end{equation}
we have $c_{\pm p}=5$, while for $v\ne\pm p$,
\begin{equation}
 c_v=5+\kappa\left(5-\frac{s^2}{2}\right)>5.
\end{equation}
For $v\ne\pm p$, since $-3\leq s\leq3$,
\begin{align}
 q_v(t)-q_p(t)&=2(3-s)t+(c_v-5)>0&& (t\geq0),\label{eq:right-envelope}\\
 q_v(t)-q_{-p}(t)&=-2(s+3)t+(c_v-5)>0&& (t\leq0).
\end{align}
The remaining comparison is $q_{-p}(t)-q_p(t)=12t$.
Thus $p$ is uniquely nearest for $t>0$, $-p$ for $t<0$, and exactly
$\{p,-p\}$ is nearest at zero.  These inequalities cover the entire
unbounded line and exclude every other cell.

The middle cell is locally the bisector $x\cdot p=0$.  Its normal $p$ has
dot product 3 with the line direction $(0,1,1)$, so the crossing is
transverse.  The projected parameters of $p,-p$ are $3/2,-3/2$, and zero is
strictly inside their segment.  The two singleton projections lie strictly
inside their corresponding rays.  Hence the pair satisfies every GNS
generic-intersection and boundary clause, all three cells are effective, and
\begin{equation}\label{eq:line-sharp}
 \sharp_L^0=1,\qquad\sharp_L^1=2.
\end{equation}
Here the genericity hypothesis is the generic-intersection condition
in GNS Conjecture~1.8(b), as in their Theorem~1.7(b), rather than
global genericity of the site configuration. The configuration
$\{B_v\}$ is not globally GNS-generic: for example,
$(-18009/8008000,0,0)$ is a Voronoi vertex with six nearest sites.
It lies outside $L$ and does not affect the three cells met by $L$.

\subsection{Two exact restricted minima}

The restricted squared distances $q_v$ fall into the following classes:
\begin{center}
\begin{tabular}{@{}ccc@{}}
\toprule
$s$ & $c$ & multiplicity\\
\midrule
$3$&$5$&1\\
$3$&$5+\kappa/2$&1\\
$2$&$5+3\kappa$&4\\
$1$&$5+9\kappa/2$&6\\
\bottomrule
\end{tabular}
\end{center}
together with the four rows obtained by replacing $s$ by $-s$.  Therefore
\begin{equation}
 V_L(t)=\sum_{(s,c)}m_{s,c}(2t^2-2st+c)^{-1/2},
\end{equation}
\begin{equation}
 V_L'(t)=\sum_{(s,c)}m_{s,c}
 \frac{s-2t}{(2t^2-2st+c)^{3/2}},
\end{equation}
and, after simplifying the numerator,
\begin{equation}\label{eq:line-second}
 V_L''(t)=\sum_{(s,c)}m_{s,c}
 \frac{8t^2-8st+3s^2-2c}{(2t^2-2st+c)^{5/2}}.
\end{equation}

The exact rational interval program encloses square roots between adjacent
rationals with denominator $10^{90}$.  It proves
\begin{equation}
 V_L'(0.62)<0<V_L'(0.63)
\end{equation}
and
\begin{equation}\label{eq:line-convex}
 V_L''([0.62,0.63])\subset
 [0.0467701569412027,0.227377521578737]\subset(0,\infty).
\end{equation}
There is a unique nondegenerate minimum in $(0.62,0.63)$.  Central symmetry
makes $V_L$ even and gives another in $(-0.63,-0.62)$.  Thus
\begin{equation}\label{eq:line-violation}
 a^0_{1/2,L}\geq2>1=\sharp_L^0.
\end{equation}
This disproves the proper-line formulation without using the degenerate
center of the regular truncated-octahedron potential.

\subsection{Suspension and exact dimension classification}

For $d\geq1$, work in
\begin{equation}
 \R^{d+2}=\R^3\times\R^{d-1},\qquad
 L_d=L\times\R^{d-1},
\end{equation}
and replace every site $B_v$ by $(B_v,0)$.  Every ambient Voronoi cell is
$S\times\R^{d-1}$.  Nearest-site sets and codimensions are unchanged,
and transversality follows by taking products of the tangent spaces with
$\R^{d-1}$.  For the lifted nearest sites, the projected
hull is exactly
\begin{equation}\label{eq:suspension-projected-hull}
 \conv\{\pi_{L_d}(B_i,0):B_i\in\NS(S)\}
 =\conv\{\pi_LB_i:B_i\in\NS(S)\}\times\{0\}.
\end{equation}
Its relative boundary is likewise the original relative boundary times
$\{0\}$.  These identities preserve relative effectiveness and
relative-boundary avoidance exactly.  The unique effective
zero-dimensional intersection becomes the unique effective intersection
of dimension $d-1$:
\begin{equation}
 \sharp_{L_d}^{d-1}=1.
\end{equation}

On $L_d$, write a point as $(t,z)$.  Then
\begin{equation}
 W(t,z)=\sum_v(q_v(t)+|z|^2)^{-1/2}.
\end{equation}
At either certified line minimum $(t_*,0)$, the Hessian is
\begin{equation}
 \operatorname{diag}\left(V_L''(t_*),
 -\left(\sum_vq_v(t_*)^{-3/2}\right)I_{d-1}\right).
\end{equation}
It has index $d-1$.  Hence
\begin{equation}
 a^{d-1}_{1/2,L_d}\geq2>1=\sharp_{L_d}^{d-1}.
\end{equation}
Since $\dim L_d=d$ and its codimension is two, this gives a proper
counterexample for every $d\geq1$.

If $\dim L=0$, then under the standard zero-manifold Morse convention a
generic point meets one full-dimensional cell, which is effective relative
to the point.  If the potential is defined there, the restriction is a
function on a singleton; its empty Hessian is nonsingular, so the point is
one nondegenerate critical point of index zero.  If the point is a charge
site, the restricted domain is empty.  Thus
$a_{\alpha,L}^0\leq\sharp_L^0=1$.  This proves the stated
zero-dimensional classification in Theorem~\ref{thm:main}.

\begin{remark}[Hyperplanes]
Embedding the ambient sites as $(A_i,0)$ and taking
$H=\R^3\times\{0\}\subset\R^4$ transfers the part-(a) violation unchanged.
If one additionally demands sites off $H$, move every site to the common
rational height $h\ne0$.  Voronoi comparisons and relative effectiveness
remain products, and for all sufficiently small $h$ the 18 nondegenerate
restricted points persist with their indices.  Thus that extra reading also
fails.
\end{remark}

\section{Surviving bounds}\label{sec:corrected}

\begin{theorem}[Configuration-dependent large exponent]\label{thm:large}
For each fixed GNS-generic positive configuration, there is an $\alpha_0$
depending on it such that, for $\alpha\geq\alpha_0$, ambient critical points
correspond to effective positive-codimension cells and their indices equal
the cell dimensions. The analogous fixed-pair relative theorem holds, with the local
maxima associated with effective full-dimensional cells whose sites lie off $L$ included.
\end{theorem}

This is the fixed-configuration conclusion of GNS Theorem~1.7,
with the full-dimensional relative case treated in their
Section~2.5 included. Their proof supplies a configuration-dependent
threshold, not a uniform one.  No universal
explicit threshold follows from their proof.

\begin{proposition}[Extreme-index exclusions]
For positive unit charges in $\R^n$,
\begin{equation}
 \Delta |x-A|^{-2\alpha}
 =2\alpha(2\alpha+2-n)|x-A|^{-2\alpha-2}.
\end{equation}
Consequently,
\begin{equation}
\begin{array}{c|c}
\alpha<(n-2)/2&a_\alpha^0=0,\\
\alpha=(n-2)/2&a_\alpha^0=a_\alpha^n=0,\\
\alpha>(n-2)/2&a_\alpha^n=0.
\end{array}
\end{equation}
\end{proposition}
\begin{proof}
The potential is strictly superharmonic, harmonic, or strictly subharmonic,
respectively.  The strong minimum/maximum principle excludes the listed
extrema.  Restriction to an affine subspace with off-subspace sites need not
preserve this sign, which is why the line example can have minima at the
ambient harmonic exponent.
\end{proof}

For a line restriction, grouping identical projected centers and heights
and deleting zero coefficient sums gives
\begin{equation}
 F(t)=\sum_{r=1}^m c_r((t-a_r)^2+b_r^2)^{-\alpha}.
\end{equation}
If $m\geq1$, Duan~\cite[Theorem~2.1]{DuanLine} and
Oliveira~\cite[Theorem~1]{OliveiraLine} give at most $2m-1$ distinct
critical points in the domain for every $\alpha>0$; Oliveira writes
$p=2\alpha$.  Duan also counts analytic multiplicity when all surviving
heights are positive; with zero heights, the bound counts distinct points
globally across the punctured domain.  If $m=0$, the restriction is
identically zero.  These results concern the nonconstant form of GNS
Conjecture~1.9 and include EFO's positive-charge Conjecture~3~\cite[p.~3]{EFO}, not the
index-by-index relative inequality of Conjecture~1.8(b).  For a positive
Morse restriction with all $b_r>0$, the extrema alternate, beginning and
ending with a maximum, so
\begin{equation}
 a_L^1=a_L^0+1,
 \qquad a_L^1\leq m,
 \qquad a_L^0\leq m-1.
\end{equation}
This sharp total bound in terms of the number of distinct grouped terms does not imply the false componentwise
relative Voronoi bound. In the line example $m=8$, so the total bound is 15 and allows at most 
seven local minima, whereas $\sharp_L^0=1$.

Writing $N_\ell(n,\alpha)$ for the maximal number of critical points over
all nondegenerate configurations with $\ell$ variable real charges, GNS
Theorem 1.5 also gives the coarse uniform estimate
\begin{equation}
 N_\ell(n,\alpha)\leq4^{\ell^2}(3\ell)^{2\ell}
\end{equation}
for nondegenerate configurations.  No universal additive or multiplicative
repair, configuration-independent $\alpha_0$, or persistence monotonicity
theorem is proved here.  The examples force any multiplicative replacement
to allow at least $9/7$ in ambient index one and at least 2 in relative index
zero.

Arathoon, Ball, and Kvalheim~\cite[Theorem~1]{ABK} disprove Maxwell's
bound using five positive charges of unequal strengths. Their
example is not a direct counterexample to the unit-charge
inequalities considered here.

\section{Independent verification}\label{sec:audit}

Two further algorithms reconstruct the certificates in addition to the primary generators.

\begin{enumerate}
\item \path{independent_verify.py} repeats the site coordinates, solves exact
  rational circumcenter equations, and compares squared distances directly
  rather than trusting the oriented in-sphere convention.  It then uses a
  320-bit explicit-adjugate interval-Newton operator~\cite[Theorem~13.2]{Rump} rather than the primary
  fixed-preconditioner Krawczyk operator.  It reproduces all Delaunay and
  effectiveness counts and all 18 index-one boxes.
\item \path{independent_adversarial_verifier.py} performs a second exact
  circumcenter reconstruction and a 512-bit Krawczyk replay with newly
  constructed rational preconditioners.  It compares outward-rounded image
  endpoints directly with the exact rational box endpoints.  It also
  reconstructs the complete line envelope and checks $V_L''>0$ on ten
  rational subintervals; its smallest lower bound is greater than $0.10609$.
\end{enumerate}

\section{Reproduction and limitations}\label{sec:reproduction}

From the project root, run
\begin{verbatim}
python3 research/certification/exact_voronoi_certificate.py
python3 research/certification/critical_point_certificate.py
python3 research/certification/verify_relative_line_counterexample.py
python3 research/certification/independent_verify.py
python3 research/certification/independent_adversarial_verifier.py
\end{verbatim}
The primary run used Python 3.13.9, python-flint 0.9.0, Arb at 256 bits, and
mpmath 1.3.0 only for noncertifying center refinement and preconditioner
construction.  Independent runs use
320- and 512-bit Arb.  No random seed is used in verification.  The JSON files contain all
sites, Delaunay cells, effective lists, rational boxes, interval images,
Hessian determinant enclosures, and exact line bounds.  In the relative-line
JSON, the top-level decimal intervals are outward-rounded displays of the exact rational bounds; the
certificate uses the fractions in \texttt{exact\_rational\_bounds}.
Whole-file JSON hashes reproduce byte-for-byte on the recorded environment;
the JSON also records Python and platform strings, so runs on another
environment should compare the mathematical payload and theorem conclusions.
Recorded hashes and this distinction are listed in the accompanying
supplementary package's \path{SHA256SUMS} and \path{README.md}.

\paragraph{Data and code availability.}
The accompanying supplementary package contains the five verification
scripts, exact input data, rational critical-point boxes, four reference
JSON outputs, dependency versions, and reproduction instructions.  No
public repository deposit is claimed here.

The ambient calculation isolates only the 18 points needed for the strict
violation; it does not enumerate every critical point.  The displayed
rational vector is GNS-generic and all 18 counted points are Morse.  The
density argument proves a rational all-Morse vector exists in the same open
chamber but does not print that second vector.  The computation trusts exact
Python integer/rational arithmetic and Arb outward rounding; two independent
algorithms reduce implementation risk, but this is not a formally
kernel-checked proof-assistant development.  None of these limitations leaves
the truth value of either printed conjecture open.

\section{Logical summary}

\begin{center}
\begin{tabularx}{0.96\linewidth}{@{}>{\raggedright\arraybackslash}X>{\raggedright\arraybackslash}X@{}}
\toprule
Input & Consequence\\
\midrule
GNS Theorem 1.7 & Effective-cell equality for each fixed configuration on a sufficiently-large-$\alpha$ tail.\\
Rational ambient certificate: $18>14$ & Part (a) false; under the inclusive convention, part (b) also fails for $L=\R^3$.\\
Ambient hyperplane lift & Proper-hyperplane version false.\\
Rational proper-line certificate: $2>1$ & Proper-line version false.\\
Product suspension & For each $d\geq1$, a proper $d$-plane counterexample in $\R^{d+2}$.\\
Zero-dimensional lemma under the standard zero-manifold Morse convention & Relative version true for $\dim L=0$.\\
\bottomrule
\end{tabularx}
\end{center}

The large-exponent theorem and the finite-endpoint counterexamples are
compatible: bifurcations must occur before the configuration-dependent tail.

\section*{AI usage statement}
ChatGPT 6 Pro was used to assist with the exploratory search for the configuration, development of the verification scripts, 
and revision of the manuscript. Claude Fable 5.1 was also used to assist with revising the manuscript. 
The author reviewed the mathematical arguments, ran and checked the computational certificates, and takes responsibility for the content of the paper.

\end{document}